\documentclass[11pt,a4]{article}
\usepackage{graphicx}
\usepackage{enumerate}
\usepackage{hyperref}
\usepackage{color}

\usepackage{amsfonts}
\usepackage{amsmath}

\usepackage[pdflatex=true,recompilepics=true]{gastex}

\newcommand{\osd}{optimal deterministic stationary}

\newcommand\M{\ensuremath{\operatorname{Max}}}
\newcommand\m{\ensuremath{\operatorname{Min}}}

\newcommand\supp{\ensuremath{\operatorname{supp}}}

\newcommand\abs[1]{\lvert#1\rvert}

\newcommand\actstat[1]{\ensuremath{\mathbf{#1}}}
\newcommand\states{\actstat{S}} 

\newcommand\size{\operatorname{size}}

\newcommand\cv[1]{\widehat{#1}} 

\newcommand\sepst{\ensuremath{\omega}} 
\newcommand\hist{\ensuremath{\mathcal{H}}} 
\newcommand\ihist{\ensuremath{\hist^\infty}} 
\newcommand\allhist{\ensuremath{\hist^{\leq\infty}}} 

\newcommand\rew{\ensuremath{\rho}} 

\newcommand\actions{\actstat{A}} 

\newcommand{\rewards}{\ensuremath{\actstat{R}}}
\newcommand{\irewards}{\ensuremath{\actstat{R}^\infty}}

\newcommand\covprew{\ensuremath{\cv{\rew}}} 
\newcommand\covprewv[2]{\ensuremath{\covprew(#1|#2)}} 
\newcommand\arenacommand[1]{\mathbb{#1}}
\newcommand\arn{\ensuremath{\arenacommand{A}}}

\newcommand\act[1]{\ensuremath{\actf{\trans}{#1}}} 
\newcommand\actf[2]{\ensuremath{\actions_{#1}(#2)}} 

\newcommand\arena{\ensuremath{(\states,\actions,%
\trans,\tpc,\player,\rewards,\rew)}}

\newcommand\arenaprime{\ensuremath{(\states,\actions,%
\trans',\tpc',\player,\rewards,\rew')}}

\newcommand\covarn{\ensuremath{(\covstates,\actions,\covtrans,
\covtpc,\covplayer,\rewards,\covprew)}}

\newcommand\trans{\ensuremath{\actstat{T}}} 
\newcommand\covtrans{\ensuremath{\cv{\actstat{T}}}} 

\newcommand\covstates{\ensuremath{\cv{\states}}}  

\newcommand\covact[1]{\ensuremath{\covtpc(#1)}}

\newcommand\covarena{\ensuremath{\cv{\arenacommand{A}}}} 

\newcommand\borel{\ensuremath{\mathcal{B}}}
\newcommand\field{\ensuremath{\mathcal{F}}}

\newcommand\player{\ensuremath{\mathbf{P}}}   
\newcommand\covplayer{\ensuremath{\cv{\player}}} 

\newcommand\plcd[2]{\ensuremath{#1[#2]}}  
\newcommand\plc[3]{\ensuremath{#1[#2](#3)}} 
\newcommand\strseta{\ensuremath{\Sigma}} 
\newcommand\strsetb{\ensuremath{\mathcal{T}}}

\newcommand{\opt}{\sharp}

\newcommand\sopt{\ensuremath{\sigma^\opt}}
\newcommand\topt{\ensuremath{\tau^\opt}}
\newcommand\covsopt{\ensuremath{\sigma^\opt}}
\newcommand\covtopt{\ensuremath{\tau^\opt}}

\newcommand{\prof}[2]{\ensuremath{(#1\cup #2)}}

\newcommand\covtau{\ensuremath{\cv{\tau}}}

\newcommand\covsigma{\ensuremath{\cv{\sigma}}}

\newcommand{\bh}{\setminus}

\newcommand\tpc{\ensuremath{p}} 
\newcommand{\covtpc}{\ensuremath{\cv{\tpc}}}

\newcommand\transition[3]{\ensuremath{#1(#2,#3)}}

\newcommand{\tp}[2]{\ensuremath{\transition{\tpc}{#1}{#2}}} 
\newcommand{\covtp}[2]{\transition{\covtpc}{#1}{#2}}

\newcommand\pay{\ensuremath{f}}  
\newcommand\embed{\ensuremath{\pi}} 
\newcommand\pembed{\ensuremath{\widehat{\embed}}} 
\newcommand\ipembed{\ensuremath{\pembed^{-1}}} 
\newcommand\iembed{\ensuremath{\embed^{-1}}} 
\newcommand\iiembed{\ensuremath{\embed^{-1}}} 

\newcommand\RR{\ensuremath{\mathbb{R}}}
\newcommand{\NN}{\ensuremath{\mathbb{N}}}
\newcommand{\zplus}{\ensuremath{\mathbb{Z}_+}}

\newcommand\p{\ensuremath{\mathbb{P}}}
\newcommand\covp{\ensuremath{\cv{\p}}}
\newcommand\proba[4]{#1_{#2}^{#3,#4}} 
\newcommand\prb[3]{\proba{\p}{#1}{#2}{#3}} 
\newcommand\covprb[3]{\proba{\covp}{#1}{\,#2}{#3}}

\newcommand\power{\ensuremath{\mathcal{P}}} 
\newcommand\pref{\ensuremath{\preceq}} 

\newcommand\mesure{\ensuremath{\mathcal{M}}}

\newcommand\outcome{\ensuremath{\mathcal{O}}}

\DeclareMathOperator{\Outcome}{\mathbb{O}}

\usepackage{amsthm}
\newtheorem{theorem}{Theorem}
\newtheorem{proposition}[theorem]{Proposition}
\newtheorem{definition}[theorem]{Definition}

\newtheorem{lemma}[theorem]{Lemma}

\theoremstyle{definition} \newtheorem{example}{Example}

\title{Optimal deterministic stationary  strategies 
in perfect-information stochastic games with global preferences}

\author{Hugo Gimbert\footnote{CNRS, LaBRI,~\tt{hugo.gimbert@cnrs.fr}.} 
\and Wies{\l}aw Zielonka\footnote{Universit{\'e} 
Paris 7, LIAFA,~\tt{zielonka@liafa.univ-paris-diderot.fr}.}}

\begin{document}
\maketitle

\begin{abstract}
We 
examine the problem of the existence 
of \osd\ strategies
in
two-players antagonistic (zero-sum) perfect information 
stochastic games with finitely 
many states and actions.
We show that 
the existence
of  such strategies follows from the existence of \osd\
strategies for some derived one-player games.
Thus we reduce
the problem from two-player to one-player games (Markov decision
problems),  where usually it is much easier to  tackle.
The reduction is very general, it holds not only 
for all possible payoff mappings but also
in more a general situations where
players' preferences are not expressed by payoffs.
\end{abstract}

\section{Introduction}
\label{sec:intro}

Given a perfect-information zero-sum stochastic game
with a finite set of states and actions, 
the existence of deterministic and stationary optimal strategies
 is a useful property.
The existence of such simple strategies has been well-studied for several
examples of games and Markov decision processes.

For example, since there are finitely many such strategies,
computability of the values of a stochastic game
is often a direct corollary of the existence of deterministic and
stationary optimal strategies.

Of course not in every
game both players have \osd\ strategies,
this depends on the transition rules of the game (the arena)
and on the way players' payoffs  are computed (the payoff function).
Actually, for various payoff functions like the mean-payoff function,
the discounted payoff function and also parity games,
players have deterministic and stationary optimal strategies
whatever is the arena they are playing in.

We provide a result which is very useful for establishing existence
of deterministic and stationary optimal strategies: 
if for some fixed payoff function $\pay$,
players have  \osd\
 strategies 
in every \emph{one-player}
stochastic game 
then this is also the case for zero-sum \emph{two-player}
stochastic games with perfect information.

In fact we prove a more general result. We show that the existence of
\osd\
 strategies for one player games implies the
existence of such strategies for two-player games for each class of
games satisfying certain closure properties.
These closure properties are satisfied by the class of all games.
We prove that this result holds also for some subclasses
of perfect information games,
for example deterministic arenas
or arenas without cycles except self-loops.

The reduction is very general, it holds not only 
for all possible payoff mappings but also
in more a general situations where
players' preferences are not expressed by payoffs
but rather by preference orders on the set of probability measures over plays.

\section{Stochastic Perfect Information Games}
\paragraph{Notation.}
$\NN=\{1,2,3,\ldots\}$, $\zplus=\NN\cup\{0\}$.
For a set $X$, $\abs{X}$ is the cadinality of $X$, $\power(X)$ is the
power set of $X$. 
$\mesure(X,\field(X))$ will stand for the set of probability measures
 on a measurable space $(X,\field(X))$, where
 $\field(X)$ 
is a
$\sigma$-algebra of subsets of $X$.  

If $X$ is finite or countably infinite then
we always assume that $\field(X)=\power(X)$ and to simplify the
notation
 we write
$\mesure(X)$ rather than $\mesure(X,\power(X))$ to denote the set of all
probability measures over $X$. Moreover for $\sigma\in\mesure(X)$,
$\supp(\sigma)=\{x\in X\mid \sigma(x)>0\}$ will denote the support of
measure $\sigma$.




\subsection{Arenas and games}
\label{sec:games:arenas}

We consider  games that two players,
 $\M$ and $\m$, play 
on an arena
\[
\arn=\arena
\]
 consisting of the following ingredients:
\begin{itemize}
\item
 \states\ and \actions\ are finite nonempty sets  of, respectively, states and
 actions,
\item 
$\trans\subseteq\states\times\actions\times\states$ is the
 set of transitions
\item
  $\tpc : \trans \to  (0,1]$ is
 the transition probability function
which
assigns for each transition
$(s,a,s')\in\trans$
  a positive  probability $\tp{s,a}{s'}$ of 
transition from  $s$ to $s'$ if  $a$ is executed at $s$.
We extend \tpc\ to all elements of $\states\times\actions\times\states$
by setting $\tp{s,a}{s'}=0$ if $(s,a,s')\not\in\trans$.
\item
$\player : \states\to\{\m,\M\}$ is a mapping
associating with each state $s\in\states$ the player $\player(s)$ 
 controling $s$. 
\item
finally, \rewards\ a set of rewards and $\rew : \trans \to \rewards$ 
is a reward mapping assigning to each transition $(s,a,s')$ a reward
$\rew(s,a,s')$.
\end{itemize}

For each $s\in\states$, we define
$\act{s}:=\{ a\in\actions \mid \exists s'\in\states, 
(s,a,s')\in\trans \}$
to be the set
of actions  \emph{available} at $s$.
We assume that all $\act{s}$ are nonempty
and,
for each action $a\in\act{s}$,
$\sum_{s'\in\states} \tp{s,a}{s'} = 1$.

A infinite game is played by players \M\ and \m\ on \arn, 
at each stage player $\player(s)$ controlling the current state $s$ chooses 
an available action $a\in\act{s}$  which results in a transition
to a  state $s'$ (which can be equal to $s$) with probability 
$\tp{s,a}{s'}$. 
Let us note that the set of rewards can be uncountable, for example 
\rewards\ can be equal to the set \RR\ of real numbers, however since
the set of transition is finite each arena contains only finitely many rewards.

In the case where all transition probabilities are equal either to $0$ or to $1$ the arena is said to be \emph{deterministic}.
In a deterministic arena, given a state $s$ and an action $a\in \act{s}$ available at $s$ there is a unique state $t$ such that $(s,a,t)\in \trans$.

\subsection{Strategies}
\label{sec:strategies}

A \emph{finite history} of length $\abs{h}=n$ in the arena $\arn$
is a finite sequence $h=s_1a_1s_2\ldots s_{n-1}a_{n-1} s_n$
alternating states and actions such that $h$ starts and ends
in a state, contains $n$ states and 
for every $1 \leq i< n$ the triplet $(s_i,a_i,s_{i+1})$
is a transition i.e. $(s_i,a_i,s_{i+1})\in\trans$.


 The set of finite histories of length $n$ is denoted
 $\hist^n(\arn)$  and
$\hist(\arn):=\cup_{i=1}^\infty \hist^n$  is the set of all  finite histories.
When $\arn$ is clear from the context we simply write $\hist$  instead of
 $\hist(\arn)$.

Let $\hist_{\M}$ be the subset of $\hist$ 
consisting of  finite histories with the last state
controlled by player \M. 

A strategy of player \M\ is a mapping 
$\sigma : \hist_{\M}\to \mesure(\actions)$
which assigns to each
$h\in\hist_{\M}$ a probability
measure over actions. 
We write $\plc{\sigma}{h}{a}$ for the probability that 
$\sigma$
assigns to action
$a$ for $h\in\hist_{\M}$. We assume that only actions available in
the last state of $h$ can be chosen, i.e.
$\supp(\sigma(h))\subseteq \act{s}$ where $s$ is the last state of $h$.

Strategy $\sigma$ 
is said to be \emph{deterministic}   
if, for each $h\in\hist_{\M}$ there is an action $a\in\actions$
such that
$\plc{\sigma}{h}{a}=1$.
We
can identify the deterministic strategies of player \M\ with the  mappings 
$\sigma : \hist_{\M} \to \actions$ such that
for  $h\in\hist_{\M}$,
$\sigma(h)$ is the action selected by $\sigma$ if the current finite history
is $h$.

A strategy $\sigma$ is \emph{stationary} 
if, for all $h\in\hist_{\M}$,
$\plc{\sigma}{h}{a}=\plc{\sigma}{s}{a}$, where $s$ is the last state
of $h$.

Deterministic stationary strategies of  $\M$ can be seen as 
mappings  from the set $\player^{-1}(\M)$ of states
controlled by \M\ to the set of actions
such that $\plcd{\sigma}{s}\in\act{s}$ for $s\in\player^{-1}(\M)$.

Strategies for player \m\ (deterministic, stationary or general)
are defined \emph{mutatis mutandis}.

$\strseta(\arn)$ and $\strsetb(\arn)$
will stand for the sets of strategies of \M\ and \m\ respectively,
and we use $\sigma$ and $\tau$ 
 (with subscripts or superscripts if necessary) to denote 
the elements of $\strseta(\arn)$ and $\strsetb(\arn)$ respectively.

A strategy profile $(\sigma,\tau)\in\strseta(\arn)\times\strsetb(\arn)$ 
(consisting of  strategies for each player)
defines a mapping $\prof{\sigma}{\tau} : \hist \to
\mesure(\actions)$,
\[
\plc{\prof{\sigma}{\tau}}{h}{a}= \begin{cases}
    \plc{\sigma}{h}{a} & \text{if the last state of $h$ is controlled by
      player \M,}\\
    \plc{\tau}{h}{a} & \text{if the last state of $h$ is controlled by
      player \m}.
\end{cases}
\]

The set of infinite histories $\ihist(\arn)$ consists of infinite sequences
$h=s_1a_1s_2a_2\ldots$ alternating states and actions
such that $(s_i,a_i,s_{i+1})\in \trans$ for every $i \geq 1$.
Again we write $\ihist$ instead of $\ihist(\arn)$ when $\arn$ is clear from the context.



For a finite history $h=s_1a_1s_2\ldots s_n$ by $h^+$ we denote the
cylinder generated by $h$ which consists of all infinite histories
having prefix $h$.
An initial state $s\in\states$ 
and a strategy profile $(\sigma,\tau)$ 
determine a probability measure 
$\prb{\arn,s}{\sigma}{\tau}\in\mesure(\ihist,\field(\ihist))$, where
$\field(\ihist)$ is the  $\sigma$-algebra
generated by the set of all cylinders.
When it is clear from the context,
we remove the arena from this notation 
and simply write 
$\prb{s}{\sigma}{\tau}$.

Given an initial state $s_1$ and strategies $\sigma$, $\tau$ of
players \M\ and \m\ we define the probability
$\prb{s}{\sigma}{\tau}$ of $h^+$:

\begin{multline}
\prb{s}{\sigma}{\tau}(h^+)=I_{s=s_1}\cdot
 \plc{\prof{\sigma}{\tau}}{s_1}{ a_1}\cdot
 \tp{s_1,a_1}{s_2}\cdot
\\
\plc{\prof{\sigma}{\tau}}{s_1a_1s_2}{ a_2}\cdot
\tp{s_2,a_2}{s_3}\cdot
\cdots \\
\plc{\prof{\sigma}{\tau}}{s_1a_1s_2\ldots s_{n-1}}{ a_{n-1}}
\cdot \tp{s_{n-1},a_{n-1}}{s_n}
\label{def:prob}
\end{multline}
where $I_{s=s_1}$ is the indicator function equal to $1$ is $s=s_1$ and
$0$ otherwise.

By
 the Ionescu Tulcea theorem \cite{neveau}
there exists a unique probability measure 
$\prb{s}{\sigma}{\tau}\in\mesure(\ihist,\field(\ihist))$
satisfying \eqref{def:prob}.
Moreover the support of this probability measure
is the set 
$\ihist(\arn,s,\sigma,\tau)$
of infinite histories whose every finite prefix
has positive probability i.e. 
\[
\ihist(\arn,s,\sigma,\tau)=
\left\{
s_1a_1s_2\cdots \in \ihist(\arn) \mid 
\forall n, \plc{(\sigma \cup \tau)}{s_1a_1s_2\cdots s_n}{a_n}>0\right\}
\enspace.
\]

\subsection{Players  preferences}
\label{sec:values}

We extend the reward mapping $\rew$
to  finite and infinite histories: 
for $h=s_1a_1s_2a_2s_3a_3s_4\ldots\in\allhist$, we set
$\rew(h)=\rew(s_1,a_1,s_2)\rew(s_2,a_2,s_3)\rew(s_3,a_3,s_4)\ldots$.


We assume that the set $\irewards$ of infinite reward sequences is
endowed with the product $\sigma$-algebra
$\field(\irewards)
:=\otimes_{i=1}^\infty\power(\rewards)$
Then the mapping $\rew$ defined above is a measurable mapping from
$(\ihist(\arn),\field(\ihist(\arn))$ to
$(\irewards,\field(\irewards))$.

This implies that, for each probability measure 
$\p\in\mesure(\ihist,\field(\ihist))$ with support
$\supp(\p)$ included in $\ihist(\arn)$, the mapping
$\rew$ induces a probability measure 
$\rew\p\in\mesure(\irewards,\field(\irewards))$
such that,
for $U\in\field(\irewards)$,
\[
\rew\p(U):=\p(\rew^{-1}(U))\enspace.
\]

We assume that  players are interested only in  infinite
sequence of
rewards obtained during the play. This leads to the following definition.
\begin{definition}[Outcomes and preference relations]
\label{def:outcome}
Fix an arena $\arn$.
For a given strategy profile $(\sigma,\tau)$ and an initial state
$s$ the \emph{outcome} of the game
is
the probability measure
\[
\Outcome(\arn,s,\sigma,\tau)=
\rew\prb{\arn,s}{\sigma}{\tau}
\in\mesure(\irewards,\field(\irewards))\enspace
\]
Given two set of strategies $\strseta \subseteq \strseta(\arn)$
and $\strsetb \subseteq \strsetb(\arn)$
We set 
\[
\outcome_s(\arn, \strseta, \strsetb):=\{\Outcome(\arn,s,\sigma,\tau) \mid
\text{$(\sigma,\tau)\in \strseta\times\strsetb$}
\} 
\]
to be the set of outcomes in $\arn$ starting at $s$ and using strategies
from $\strseta$ and $\strsetb$ and
\[
\outcome(\arn,\strseta,\strsetb)
=\bigcup_{s\in\states} \outcome_s(\arn,\strseta,\strsetb)\enspace
\]
the set of all possible outcomes using strategies
from $\strseta$ and $\strsetb$.
%

A \emph{preference relation in $\arn$ with strategies in $\strseta$ and $\strsetb$} is a reflexive and transitive 
binary relation \pref\
over $\outcome(\arn,\strseta,\strsetb)$
\end{definition}

%


A preference relation \pref\ is \emph{total} if
for 
all outcomes $\p_1,\p_2\in\outcome(\arn,\strseta,\strsetb)$, 
either
$\p_1 \pref
\p_2$ or
$\p_2 \pref
\p_1$. Usually naturally arising preference relations  are total
but this assumption is not necessary to formulate and prove our main
result.

\subsection{Games and optimal strategies}

A \emph{game} is a tuple $\Gamma=(\arn,\strseta,\strsetb,\pref)$ composed of an arena,
some set of strategies for the players in $\arn$
and a preference relation \pref\ in $\arn$ with strategies in $\strseta$ and $\strsetb$.
The aim of \M\ is
to maximize the obtained outcome with respect to \pref.
We consider  only zero-sum games where 
the preference relation of
player \m\  is the inverse of \pref.

A strategy 
$\sigma\in\strseta$ is a \emph{best response} to a strategy
$\tau\in\strsetb$ in $\Gamma$  if for each state $s$ and  
each strategy $\sigma'\in\strseta$ of \M,
\[
\Outcome(\arn,s,\sigma',\tau)
\pref
\Outcome(\arn,s,\sigma,\tau)\enspace.
\] 
Symmetrically, a best response 
of player \m\ 
to a
strategy $\sigma$ of \M\ is a strategy $\tau$ such that
 for all strategies $\tau'\in\strsetb$ of \m,
\[
\Outcome(\arn,s,\sigma,\tau)
\pref
\Outcome(\arn,s,\sigma,\tau')\enspace.
\] 

A pair of strategies $\sopt\in\strseta,\topt\in\strsetb$ is \emph{optimal} 
if \sopt\ is a best response to
\topt\ and \topt\ a best response to \sopt,
i.e.
if, for each state $s$ and all strategies $\sigma\in\strseta$, 
$\tau\in\strsetb$,
\begin{equation}
\Outcome(\arn,s,\sigma,\topt)
\pref
\Outcome(\arn,s,\sopt,\topt)
\pref
\Outcome(\arn,s,\sopt,\tau)\enspace.
\label{eq:optimalityi}
\end{equation}

It is an elementary exercise to check that if 
$(\sopt_1,\topt_1)$ and $(\sopt_2,\topt_2)$ are pairs of optimal strategies
then $(\sopt_1,\topt_2)$ and $(\sopt_2,\topt_1)$ also are.
As a consequence we say that a single strategy $\sopt\in\strseta$ itself is optimal whenever it belongs
to some pair of optimal strategies $(\sopt,\topt)$, and similarly for $\topt\in\strsetb$.

\subsection{Deterministic games}

For some applications, it is natural to require both arenas and strategies
to be deterministic.
In this case, it is enough to express the preferences 
of the players between infinite sequences of rewards because the probability measures defined
by strategy profiles are Dirac measures over $\irewards$.
For this purpose we define \emph{deterministic preference relations} as follows.

\begin{definition}
\label{def:detoutcome}
Fix a deterministic arena $\arn$.
For a given profile $(\sigma,\tau)$ of deterministic strategies
and an initial state
$s$ there is a unique infinite history $h(s,\sigma,\tau)=s_1a_1s_2a_2\ldots \in \ihist$ such that
$s=s_1$ and for every $n$,
\begin{align*}
& \prof{\sigma}{\tau}(s_1a_1\ldots a_{n-1}s_n) = a_n\enspace\text{ and}\\
& (s_n,a_n,s_{n+1})\in\trans\enspace,
\end{align*}
and $\rew(h(s,\sigma,\tau))\in\irewards$ is called 
the deterministic outcome of the game.
%
%
A \emph{deterministic preference relation in $\arn$} is a reflexive and transitive 
binary relation \pref\
over the set 
of deterministic outcomes in $\arn$.
\end{definition}

A \emph{deterministic game} is a tuple $\Gamma=(\arn,\strseta,\strsetb,\pref)$ composed of a deterministic arena,
the sets  $\strseta$ and $\strsetb$ of deterministic strategies in $\arn$
and a deterministic preference relation \pref\ of player \M.
In a deterministic game
strategies $\sopt,\topt$ are optimal
if for each state $s$ and all deterministic strategies $\sigma\in\strseta$, 
$\tau\in\strsetb$,
\begin{equation}
\rew(h(s,\sigma,\topt))
\pref
\rew(h(s,\sopt,\topt))
\pref
\rew(h(s,\sopt,\tau)).
\label{eq:detoptimalityi}
\end{equation}

\section{Examples}
\label{sec:payoff}

Specifying  players' preferences  by means of  preference relations
over measures
allows us 
to cover a wide range of optimality criteria. We illustrate this
flexibility with four examples.

In most application the preferences  are rather defined  by means of a payoff
mapping.

A payoff mapping is a measurable mapping
$f$ from the set $(\irewards,\borel(\irewards))$
of infinite reward sequences to the set $(\RR,\borel(\RR))$
of real numbers equipped with the $\sigma$-algebra of Borel sets.

For each outcome $\rew \prb{s}{\sigma}{\tau}$ 
we write $\rew \prb{s}{\sigma}{\tau}(f)$ for the
expectation of $f$, 
\[
\rew \prb{s}{\sigma}{\tau}(f):=\int_{\irewards} f d\p\]
 where
$\p:=\rew\prb{s}{\sigma}{\tau}$.
We assume that $f$ is integrable for all outcome measures
$\rew\prb{s}{\sigma}{\tau}$.

We say that a preference relation \pref\ is induced by a payoff
mapping $f$ if
for any two outcomes $\p_1$ and $\p_2$, $\p_1\pref \p_2$ iff
$\p_1(f)\leq\p_2(f)$.

Mean-payoff games and parity games are two well-known examples of
games with preferences  induced by payoff mappings.

\begin{example}[Mean-payoff games]
\label{ex:mean}
A mean-payoff game is a game played on  arenas 
equipped with  a reward mapping 
$\rew : \trans \to \RR$
and the payoff of an infinite reward sequence
$r_1r_2r_3\ldots$  is given by 
$\limsup \frac{\sum_{i=1}^n r_i}{n}$.

\end{example}
 
\begin{example}[Parity games]
\label{ex:parity}
The class of games with many applications in computer science and
logic is the class of parity games \cite{dagstuhl}.
These games are played on arenas endowed with a priority mapping
$\beta : \states \to \zplus$ and the payoff for an infinite history 
is either $1$ or $0$  depending on whether 
$\limsup_i \beta(s_i)$, the maximal priority visited infinitely
often,
 is odd or even, where $s_i$ is the state visited at stage $i$.
Again the aim of players \M\ and \m\ is, respectively, 
to maximize/minimize the probability
\[
\prb{s}{\sigma}{\tau}( \limsup_n \beta(s_n) \text{ is even})\enspace.
\]
\end{example}

The next example is a variant of parity games
which is positional in deterministic arena
but in general not in stochastic arenas.
\begin{example}[Simple parity games]
\label{ex:simpleparity}
A simple parity game is played in a parity game arena,
the aim of players \M\ and \m\ is, respectively, 
to maximize/minimize the probability
\[
\prb{s}{\sigma}{\tau}( \sup_n \beta(s_n) \text{ is even})\enspace.
\]
\end{example}

We continue with two  examples where the preference relation is not
induced by a payoff over infinite reward sequences.
The first is a well-known variant of mean-payoff games
of Example~\ref{ex:mean}.
\begin{example}[Mean-payoff games]
\label{ex:mp}
The arena is equipped with a real valued reward mapping $\rew :
\trans \to\RR$
exactly like in Example~\ref{ex:mean}.
Let $f_n = \frac{\sum_{i=1}^n \rew(s_i,a_i,s_{i+1})}{n}$ be the
mean-payoff over first $n$ periods.
For outcomes $\p_1$ and $\p_2$  we set $\p_1\pref \p_2$ if
$\limsup_n \p_1( f_n )\leq \limsup_n \p_2( f_n)$.
\end{example}

The next example is variant of the overtaking optimality criterion:
\begin{example}[Overtaking]
\label{ex:over}
Let $\rew : \trans \to \RR$ 
be a real valued reward mapping. 
For outcomes $\p_1$ and $\p_2$ we set
$\p_1\pref\p_2$ if there exists $n$ such that for all $k\geq n$,
$\p_1( \sum_{i=1}^k \rew(s_i,a_i,s_{i+1}) ) \leq
\p_2( \sum_{i=1}^k \rew(s_i,a_i,s_{i+1}) )$. This preference
relation is not total.
\end{example}

In the first three examples we associated with each outcome a real
number which allows us to order the outcomes. 
These real numbers can be used also to quantify how  much one outcome
is better than another one for a given player
but the question ``how much better'' is irrelevant when we are
interested in optimal strategies. Moreover, as
Example~\ref{ex:over}
shows,
for  some preference relations it is
difficult to  define a corresponding payoff mapping.

\section{From one-player to two-player games: the case of deterministic games and perfect-information stochastic games }
\label{sec:framework}

Our main result relies on three notions: one player games, 
subarenas and coverings. Before stating our theorem in its full generality, we provide in this section a weaker form
of Theorem~\ref{theo:main} which does not make use of the notions of subarenas and coverings and is enough to cover most applications.

\subsection{One-player arenas}
\label{sec:oneplayer}

We say that arena $\arn=\arena$
is  a \emph{one-player arena} controlled by
$\M$ 
if each state $s$ controlled by player \m\ has only one available action. 
But for states  with one available action 
 it is essentially irrelevant which player controls them
and we can as well assume that all states of such an arena \arn\ 
are controlled by \M.
Therefore  games on one-player arenas controlled by
player \M\ are nothing else but Markov decision processes where the unique
player \M\ wants to maximize (relative to \pref) 
the resulting outcome.

One-player arenas controlled by player \m\ are defined in a symmetric
way and they can be identified with Markov decision processes where
the aim of the
unique player \m\ is to minimize the outcome relative to \pref.

A one-player game 
is a game on a one-player arena \arn\ controlled either by \M\ or by \m.

\subsection{Specializations of the one-to-two theorem}
\label{sec:weak}

Our main result reduces the problem of the existence of \osd\
strategies in two-player games to the same problem for one-player games. 

For the reader interested only in using our theorem for a particular class of games,
 it may be sufficient to make use of one of the two following weak forms of our result, 
 which adresses specifically the cases of deterministic games
 and of perfect-information stochastic games.
 
The general statement and its proof may be found in the next section.

We start with a specialization of our theorem to the
class of deterministic games.

\begin{theorem}[One-to-two theorem for deterministic games]
\label{theo:deterministic}
Let $\rewards$ be a set of rewards and $\pref$ a reflexive and  transitive relation on $\irewards$. A deterministic game $G=(\arn_G,\pref_G)$ with rewards in $\rewards$ is said to be \emph{compatible} with $\pref$ if $\pref_G$ is the restriction of $\pref$ to the set of deterministic outcomes of $G$.

 Assume that \osd\ strategies exist in every deterministic one-player games compatible with $\pref$.
 Then \osd\ strategies exist in every deterministic two-player games compatible with $\pref$.
 \end{theorem}

Remark that we do not require the transitive relation $\pref$ to be total, in other words $\pref$ is simply assumed to be a partial preorder. This degree of generality is natural because in this paper all arenas are assumed to be finite, hence any infinite history arising in a game goes through finitely many different transitions and generates a sequence of rewards that takes finitely many different values in $\rewards$. Thus it is useless  to define $u \pref v$ when either $u$ or $v$ takes infinitely many different letters from $\rewards$.

When applied to the class of all perfect-information two-player games, our theorem specializes as follows. 

\begin{theorem}[One-to-two theorem for games with perfect-information]
\label{theo:weak}
Let $\rewards$ be a set of rewards and $\pref$ a reflexive and transitive relation on the set $\mesure(\ihist,\field(\ihist))$ of probability measures on $\ihist$ equipped with the  $\sigma$-algebra
$\field(\ihist)$ generated by the set of all cylinders.

A game $G=(\arn_G,\pref_G)$ with rewards in $\rewards$ is said to be \emph{compatible} with $\pref$ if $\pref_G$ is the restriction of $\pref$ to the set of outcomes of $G$.

 Assume that \osd\ strategies exist in every deterministic one-player games compatible with $\pref$.
 Then \osd\ strategies exist in every deterministic two-player games compatible with $\pref$.
\end{theorem}

%

Again remark that $\pref$ is not assumed to be total, and actually we only need to compare probability measures that arise as outcomes in an arena $\arn$ with rewards $\rewards$
(i.e. elements of $\outcome(\arn)$). In particular, it is enough to define $\pref$ for probability measures whose support is included in the set of histories that take finitely many different values in $\rewards$.

Although the formulations of the two specializations of our theorem are quite close,
they apply to different situations.
The first one adresses deterministic games,
where players are restricted to deterministic strategies and the preference relations
are defined only on infinite sequences of rewards,
while the second one adresses the full class of stochastic games where players use 
behavioural strategies and their preferences are expressed with respect to probability measures
on  infinite reward sequences.
As a matter of fact there is no way, at least to our knowledge,
to deduce the second theorem from the first one nor the contrary.

\subsection{Examples revisited}
\label{sec:revisited}

%
%

Bierth~\cite{bierth} shows that one-player games 
of Example~\ref{ex:mean}
have \osd\
strategies.
A more readable proof of this fact is given in
\cite{Neyman:2003b}. 

One-player parity games of Example~\ref{ex:parity} have also
\osd\  strategies \cite{deAlfaro:1997}, thus the
same holds for two-player parity games (the latter fact was proved by
several authors but we can see now that this is a consequence of the
result for one-player games).

Unlike parity games, one-player \emph{simple} parity games
are \emph{not} positional, however they become positional when played on deterministic arenas.
This is discussed in the next subsection.

%

One-player games of Example~\ref{ex:mp} have \osd\
 strategies \cite{derman}. Theorem~\ref{theo:weak}
allows us to deduce that the same holds for two-player games
with perfect information.

One-player games of Example~\ref{ex:over} do not have 
\osd\ strategies. For example if the game is made of two simple cycles of length $4$ on the initial state, labelled with $0,1,1,0$
on one hand and $1,0,0,1$ on the other hand then neither cycle dominates the other.

\subsection{Application to deterministic games}
\label{sec:deterministic}

A deterministic game is a game played on a deterministic arena with
the additional constraint that
both players can use only deterministic strategies.

Even if deterministic games are usually much simpler to analyze 
than their stochastic counterparts, this does not mean 
that they are always easy.  
They are also interesting by themselves for at least two reasons.
First of all
deterministic games  prevail 
in computer science applications related to automata theory, logic and 
verification \cite{dagstuhl}.
The second reason is that, concerning our main problem -- the existence
of \osd\ strategies, deterministic games can differ from
their stochastic counterparts.

It turns out that for the payoff mapping $f$ associated to simple parity games,
all deterministic games with payoff $f$ have \osd\ strategies but the
same is not true for simple stochastic parity games.
A suitable
  example was
shown to us by Florian Horn, see 
Figure~\ref{fig:florian}.

\begin{figure}[ht]
\begin{center}
\includegraphics[scale=0.7]{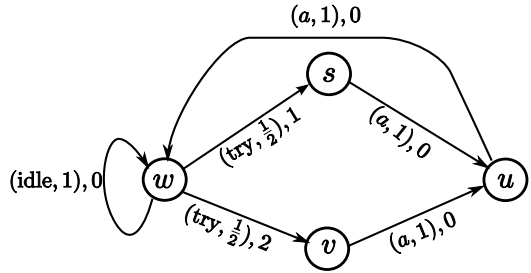}
\end{center}
\caption{A simple parity game with all states controlled by player \M.
Each transition $t$ is labelled by a triple 
$((b,p),r)$ where $b$ is an action, $r$ the reward and $p$ is the transition
probability of $t$. There is only one randomized action $\mathbf{try}$,
its execution yields with probability $1/2$ either the reward $1$ or $2$.
}
\label{fig:florian}
\end{figure}

In the game in Figure~\ref{fig:florian} only the state $w$ has two available actions
$\mathbf{try}$ and $\mathbf{idle}$.
The payoff of \M\ is determined in the
following way.
Let $r= r_0r_1r_2\cdots$ be the infinite sequence of rewards,
player \M\ obtains payoff $1$ if this sequence contains at least one occurence
of reward $1$ and no occurence of reward $2$,
otherwise \M\ gets payoff $0$. Clearly the
optimal strategy for player \M\ is to execute the action try exactly once. If the transition with reward  $1$ is taken
 then returning to $w$ he should always
execute $\mathbf{idle}$, if the transition with reward $2$ is taken then he can
do anything, his payoff will be $0$. Thus his maximal expected payoff
is $1/2$. This strategy is not stationary since player \M\ should
remember if it is the first visit to $w$ or not.

On the other hand it is easy to see that all stationary strategies
(randomized or not) yield the expected payoff $0$.

\subsection{A remark on 
one-player stochastic games with sub-mixing payoff functions}
\label{sec:mixing}
In light of Theorem~\ref{theo:weak} 
it is important to be able to decide if one-player 
games have \osd\   strategies.
A convenient way to tackle this problem 
was discovered by the first author~\cite{stacs}.

Let $f : \actions^\infty \to \RR$ be a bounded (either from below or
from above) and measurable.

We say that $f$ is prefix-independent $f$ does not depend on the
finite initial segment, i.e. $f(a_1a_2a_3\ldots)=f(a_2a_3\ldots)$ for
each infinite sequence $a_1a_2a_3\ldots\in\actions^\infty$.
And $f$ is said to be  sub-mixing if 
\[
f(a_1a_2a_3\cdots)\leq\max\{f(a_{i_1}a_{i_2}\cdots),f(a_{j_1}a_{j_2}\cdots)\} ,
\]
for each infinite sequence $a_1a_2a_3\ldots\in\actions^\infty$ of
actions and each partition of \NN\ onto two infinite sets 
$I=\{i_1<i_2<i_3<\ldots\}$, $J=\{j_1<j_2<j_3<\ldots\}$.

In \cite{stacs} it is proved that if $f$ is prefix independent 
and  sub-mixing then for
each one-player game $\Gamma=(\arn,\pref)$ 
with the preference  relation induced by $f$
and such that
 \arn\ is
controlled by player \M, this player has an \osd\
  strategy.
 
In order  to prove that one-player games controlled by player \m\ 
have \osd\  strategies it suffices to verify 
that $-f$ is sub-mixing.
Thus verifying that $f$ and $-f$ are prefix independent and
 sub-mixing can be used to prove
that two-player games with the preference relation induced by $f$ 
have

Let us note that the payoff $f$ of the parity game of
Example~\ref{ex:parity} is prefix independent and sub-mixing  and the
same holds for $-f$.

For the mean-payoff games of Example~\ref{ex:mean} 
the payoff $f$ is sub-mixing, but $-f$ is not (replacing $\limsup$ by
$\liminf$ we obtain the payoff mapping which is not sub-mixing).
However, for if players use deterministic stationary strategies then
we a get a finite state Markov chain and for the resulting outcome \p\
we have
$\limsup \frac{\sum_{i=1}^n \rew(s_i,a_i,s_{i+1})}{n} = 
\lim \frac{\sum_{i=1}^n \rew(s_i,a_i,s_{i+1})}{n}$ almost surely,
where $s_i,a_i,s_{i+1})$ is the transition at stage $i$.
This can be used to prove that one-player games controlled by
player \m\ have \osd\ strategies we can apply
Theorem~\ref{theo:weak} once again.


\section{From one- to two-player games: the general theorem}

Our main result relies on the notions of subarenas and splits of an arena on a state.

\subsection{Subarenas and subgames}

\label{sec:subarena}
An arena $\arn'=\arenaprime$
is a subarena of another arena\newline
$\arn=\arena$ if
 $\trans'\subseteq\trans$ and
 $\tpc'$ and $\rew'$
 are the restrictions of 
$\tpc$   and
   $\rew$ to 
$\trans'$.

Clearly the set  $\trans'$ of transition 
of subarena $\arn'$  satisfies the following conditions:
\begin{enumerate}[(a)]
\item
if $(s,a,s')\in\trans'$ and $(s,a,s'')\in\trans$ then
$(s,a,s'')\in\trans'$, i.e. for each state $s$ and each action $a$ either
we keep all transitions $(s,a,\cdot)$ in $\arn'$ or we remove them all,
\item
for each $s\in\states$ there exist $a\in\actions$ and
$s'\in\states$
such that $(s,a,s')\in\trans'$.
\end{enumerate}
By definition of an arena, both conditions (a) and (b) are necessary for $\arn'$ to be an arena:
(a) is necessary for the sum of transition probabilities induced by a pair of states and actions
to be $1$ and  (b) because in each state there should be at least one available action.
They are also sufficient: if $\trans'\subset\trans$ satisfies (a) and (b) then
there exists a unique subarena $\arn'$ of \arn\ having $\trans'$ as the 
set
of transitions.

If $\arn'$ is a subarena of \arn\ then all finite histories
consistent with $\arn'$ are consistent with \arn.
 
We say that a strategy $\sigma$ in $\arn$ is 
compatible with $\arn'$ if its restriction
to $\hist(\arn')$ is a strategy in $\arn'$.
By definition of a strategy, this requires
that $\sigma$ restricted to $\hist(\arn')$ 
only puts positive probability on actions available in $\arn'$,
i.e. formally for every $s_1a_2s_2\cdots s_n\in \hist(\arn')$ and $a\in \actions$,
\[
\sigma(s_1a_2s_2\cdots s_n)(a)> 0 \implies \exists t\in \states, (s_n,a,t)\in\trans'\enspace.
\]
This condition is also sufficient for $\sigma$ to be compatible with $\arn'$.

%
Conversely, every strategy $\sigma'$ in $\arn'$ defined on $\hist(\arn')$
can be extended in an arbitrary way to a strategy in $\arn$
defined on $\hist(\arn)$.
This implies the inclusion
$\outcome(\arn')\subseteq\outcome(\arn)$ of the sets of outcomes.
Therefore a preference \pref\ relation
defined on outcomes of \arn\  is also a
preference on outcomes of $\arn'$.

\begin{definition}[Subgame induced by a subarena]\label{defi:subgame}
Let $\arn'$ be a subarena of $\arn$ and 
$G=(\arn,\strseta,\strsetb,\pref)$ a game played in $\arn$.
The subgame of $G$ induced by $\arn'$ is the game $G'=(\arn',\strseta',\strsetb',\pref')$
where:
\begin{itemize}
\item[i)] $\strseta'$
is the set of strategies obtained by restricting to $\hist(\arn')$ the strategies in $\strseta$
compatible with $\arn'$,
\item[ii)] $\strsetb'$
is the set of strategies obtained by restricting to $\hist(\arn')$ the strategies in $\strsetb$
compatible with $\arn'$,
\item[iii)]
$\pref'$ is the restriction of $\pref$
to $\outcome(\arn',\strseta', \strsetb')$.
\end{itemize}
\end{definition}

\subsection{Split of an arena}

Intuitively, the split of an arena $\arn=\arena$ on a state
$\sepst\in\states$ consists in adding in the states of the arena
an extra information: the 
last action chosen in the state $\sepst$ by the player controlling $\sepst$. This memory is reset each time 
$\sepst$ is reached. The state $\sepst$ is called the \emph{separation state}.
An example is provided on Fig.~\ref{fig:splitexample}.

\newcommand{\sparn}{\cv{\arn}}
\newcommand{\spstates}{\covstates}
\newcommand{\sptrans}{\covtrans}
\newcommand{\sptpc}{\cv{\tpc}}
\newcommand{\spplayer}{\covplayer}
\newcommand{\sprew}{\cv{\rew}}
\newcommand\sparena{\ensuremath{(\spstates,\actions,\sptrans,\sptpc,\spplayer,\rewards,\sprew)}}
\newcommand{\spproj}{\pi}
\newcommand{\Spproj}{\Pi}
\newcommand{\splift}{\phi}
\newcommand{\Splift}{\Phi}

\begin{figure}[h]
\begin{gpicture}
\gasset{} 
  \put(0,0) {
 \drawpolygon[Nframe=y](-10,20)(30,20)(30,-40)(-10,-40)

  \drawline[Nframe=y,AHnb=0](-10,-30)(0,-30)(0,-40)
\node[Nframe=n](r)(-5,-35){$\arn$}
  
	\node(s)(0,0){$s$}
  \node(omega)(20,-20){$\omega$}
  
  \drawedge[curvedepth=5](omega,s){$a,\frac{1}{2}$} 
  \drawloop[ELside=l,loopangle=270](omega){$a,\frac{1}{2}$} 
  \drawedge[curvedepth=15,ELside=l](omega,s){$b,1$} 
  
  \drawloop[ELside=l,loopangle=90](s){$a,1$} 
  \drawedge[ELside=l,ELpos=40,curvedepth=5](s,omega){$b,1$} 
  }
  
    \put(65,0) {
 \drawpolygon[Nframe=y](-10,20)(50,20)(50,-40)(-10,-40)
 
  \drawline[Nframe=y,AHnb=0](-10,-30)(0,-30)(0,-40)
\node[Nframe=n](r)(-5,-35){$\sparn$}

	\node(s)(0,0){$s_a$}
  \node(omega)(20,-20){$\omega$}
	\node(sb)(40,0){$s_b$}
  
  \drawedge[curvedepth=5,color=blue](omega,s){$a,\frac{1}{2}$} 
  \drawloop[ELside=l,loopangle=270](omega){$a,\frac{1}{2}$} 
  \drawedge[curvedepth=-15,ELside=r](omega,sb){$b,1$} 

  \drawloop[ELside=l,loopangle=90](s){$a,1$} 
  \drawedge[ELside=l,ELpos=30,curvedepth=5](s,omega){$b,1$} 

  \drawloop[ELside=l,loopangle=90](sb){$a,1$} 
  \drawedge[ELside=r,ELpos=30,curvedepth=-5](sb,omega){$b,1$} 
}
 \end{gpicture}
 \caption{\label{fig:splitexample} Splitting an arena. The arena $\arn$ on the left handside has two states $\{s,\omega\}$
 and two action $\{a,b\}$, the transitions 
 and their probabilities are represented graphically, for example $\tpc(\omega,a,s)=\frac{1}{2}$.
 The split of $\arn$ on the separation state $\omega$ is the arena $\sparn$ represented on the right handside.}
\end{figure}

Formally, the split of arena $\arn$ on the separation state $\sepst$ is
the arena $\sparn=\sparena$ defined as follows.
The actions and rewards are the same in both arenas $\arn$ and $\sparn$.
A state of $\sparn$ is either the separation state $\omega$ or a pair made of another state
and an action available in the separation state,
thus
\[
\spstates=  \{\omega\} \cup  (\states \bh \{\omega\}) \times \act{\sepst}  \enspace,
\]
where $ \act{\sepst}$ denotes the set of actions available in state $\sepst$.
For every state $s \in \states \bh \{\omega\}$ and action $x\in\act{\sepst}$ we denote 
$s_x = (s,x)$. 
To make the notation uniform we set $\sepst_x=\sepst$
 i.e. $\sepst_x$ is just
an alias name for $\sepst$. There is a natural mapping
\[
\spproj : \spstates \to \states
\]
which sends $s_x$ to $s$.
The rewards and the controlling  player are the same in states $s$ and $s_x$: we set 
$\sprew = \rew \circ \spproj$ and $\spplayer = \player \circ \spproj$.
The transitions of $\sparn$ and their probabilities are inherited from $\arn$
so that the second component of a state in $\sparn$ keeps track of the last action chosen in the separation state:
\begin{align*}
&\sptrans = \{ (\sepst, x, t_x) \mid  (\sepst, x, t) \in \trans \} \cup \{ (s_x,a,t_x) \mid  (s,a,t) \in \trans \land s \neq \sepst \}\enspace,\\
& \text{for every } (s_x,a,t_y)\in \sptrans \text{ we set } \sptpc(s_x,a,t_y) = \tpc(s,a,t)\enspace.
\end{align*}

The separation state $\sepst$ separates the states
of the split $\sparn$ in the following sense.
For every action $x\in\act{\sepst}$ we denote
\[
\spstates_x = (\states \bh \{\sepst\})\times \{x\}\enspace.
\]
Then 
$\spstates$ is partitioned in $\{\sepst\}$ on one hand and the sets $(\spstates_x)_{x\in\act{\sepst}}$
on the other hand.
By construction of the split:
\begin{proposition}[Separation property]\label{prop:sep}
Every finite history of $\sparn$
starting in $\spstates_x$ and ending in $\spstates_y$ with $x\neq y$
passes through
$\sepst$.
\end{proposition}

\subsection{Projecting histories and lifting strategies}

\begin{figure}[h]
\begin{gpicture}
\gasset{Nadjust=wh,Nframe=n}
 
 \drawpolygon[Nframe=y](-35,10)(90,10)(90,-25)(-35,-25)
  \put(0,5) {
\node(A)(-10,-5){\bf Projecting histories}
  \node(C)(20,-5){$\hist(\sparn)$}
  \node(G)(20,-10){$s_1a_1s_2\cdots s_n$}
  \node(D)(60,-5){$\hist(\arn)$}
  \node(F)(60,-10){$\spproj(s_1)a_1\spproj(s_2)\cdots \spproj(s_n)$}
  \drawedge(C,D){$\spproj$} 
  }
 \put(0,-10) {
 \node(A)(-10,-5){\bf Lifting  strategies}
  \node(C)(20,-5){$\strseta(\sparn)$}
  \node(G)(20,-10){$\sigma \circ \spproj$}
  \node(D)(60,-5){$\strseta(\arn)$}
  \node(F)(60,-10){$\sigma$}
  \drawedge[ELside=r](D,C){$\Spproj$} 
  }
  
 
  \put(0,-45) {
\node(A)(-10,10){Fix $x \in\act{\sepst}$}
   \drawpolygon[Nframe=y](-35,15)(90,15)(90,-35)(-35,-35)
  \put(0,5) {
\node(A)(-10,-5){\bf Lifting histories}
  \node(C)(20,-5){$\hist(\sparn)$}
 
  \node(G)(-12,-13){$\splift_x(s) = s_x$}
  \node(Gp)(-12,-20){$\spproj \circ \splift_x  \text{ is the identity.}$}

  \node(D)(60,-5){$\hist(\arn)$}
  \node(F)(60,-10){}
  \drawedge[ELside=r](D,C){$\splift_x$} 
  }
 \put(0,-20) {
 \node(A)(-10,-5){\bf Projecting  strategies}
  \node(C)(20,-5){$\strseta(\sparn)$}
  \node(G)(20,-10){$\sigma$}
  \node(D)(60,-5){$\strseta(\arn)$}
  \node(F)(60,-10){$\sigma \circ \splift_x$}
  \drawedge(C,D){$\Phi_x$} 
  

  }
  }
%
\end{gpicture}
  \caption{\label{fig:projlift} Projections and liftings of histories and strategies and their characteristic properties.
  Strategies of $\m$ in $\strsetb(\arn)$ and $\strsetb(\sparn)$ are projected and lifted in a symmetric way.}
\end{figure}

Let $\sparn$ be the split of an arena $\arn$ on some separation state $\sepst$.

There is a natural projection $\spproj$
of finite histories in the split $\sparn$ to finite histories in the original
arena $\sparn$ and conversely a natural lifting $\Spproj$
of strategies in the original arena to strategies in the split.
These mappings are represented on Fig.~\ref{fig:projlift}. 

The projection $\spproj : \hist(\sparn) \to \hist(\arn)$ is the natural extension
to histories of the mapping $\spproj:\spstates\to\states$
which maps $\omega$ to $\omega$ and forgets the second component of other states ($\forall s_x \in \spstates, \spproj(s_x)=s$).
For every 
finite history $h=s_1a_1s_2\cdots s_n \in \hist(\sparn)$ in $\sparn$,
we set $
\spproj\left(h\right)=\spproj(s_1)a_1\spproj(s_2)\cdots \spproj(s_n).
$
That $\spproj\left(h\right)$ is a finite history in $\hist(\arn)$
is immediate from the definition of $\sptrans$.

The mapping $\Spproj:\strseta(\arn) \to \strseta(\sparn)$
 transforms a strategy $\sigma\in \strseta(\arn)$ to the strategy
$\Spproj(\sigma) = \sigma\circ \spproj$
 called the \emph{lifting} of $\sigma$
from $\arn$ to $\sparn$. Obviously $\sigma \circ \spproj$
is a strategy in $\sparn$
because
for every state $s_x\in\spstates$, 
the same actions are available in state $s_x$ of the arena $\sparn$
and  in the state $s$ of the arena $\arn$.

%

\subsection{Lifting histories and projecting strategies}

There is a canonical way to lift histories from the
original arena to the splitted arena, provided the
initial state is fixed.
\begin{proposition}
\label{lem:imageco}
Fix an action $x\in\act{\sepst}$.
For every state $s$ of $\arn$
and history $h\in\hist(\arn,s)$ starting in $s$ there exists a unique history $\splift_x(h)\in \hist(\sparn,s_x)$ such that $\pi(\phi_x(h))=h$.
F every state $s\in \states$, $\splift_x(s)=s_x$ and
in particular $\splift_x(\omega)=\omega$.
Moreover, for every for every state $s\in\states$
and
strategies
$\cv{\sigma}\in\strseta(\sparn)$
and $\cv{\tau}\in\strsetb(\sparn)$,
\begin{align}
\label{eq:projliftid}
&\spproj \circ \splift_x\text{ is the identity on } \hist(\arn)\\
\label{eq:projliftid4}
&\splift_x \circ \spproj \text{ is the identity on } \hist(\arn,s_x,\cv{\sigma},\cv{\tau})\enspace.
\end{align}
\end{proposition}

\begin{proof}
The mapping $\splift_x$ is defined inductively.
Initially,
$\splift_x$ maps $s$ to $s_x$.
Assume $\splift_x(h)$ is uniquely defined for some
$h\in \hist(\arn)$ and let $a\in\actions$ and $s\in \states$ such that $has\in \hist(\arn)$.
Let $t$ be the last state of $h$. If $t=\omega$ then there is a unique $y\in \act{\sepst}$
such that $(\omega,a,s_y)$ is a transition in $\sparn$: necessarily $y=a$.
If $t\neq \omega$ then by induction, according to property i)
there exists $b\in \act{\sepst}$ such that $t_b$ is the last state of $\splift_x(h)$.
There is a unique $y\in \act{\sepst}$
such that $(t_b,a,s_y)$ is a transition in $\sparn$: necessarily $y=b$.
Equations~\eqref{eq:projliftid} holds by definition of $\splift_x$.
Equation~\eqref{eq:projliftid4}
follows from an easy induction. 
\end{proof}

Beware that the lifting of a finite history $s_1a_1s_2\cdots s_n$ 
in $\arn$ is in general different from $\splift_x(s_1)a_1\splift_x(s_2)\cdots \splift_x(s_n)$
which actually may not even be a finite history in $\sparn$.
For example, in the arena of Fig.~\ref{fig:splitexample},
the lifting of $h=s,a,\omega,b,s,b,s,b,s $ with respect to $a$ is
$\splift_a(h)=s_a,a,\omega,b,s_b,b,s_b,b,s_b $
and not $s_a,a,\omega,b,s_a,b,s_a, b,s_a$.

The lifting of histories from $\arn$ to $\sparn$
can be used to project strategies from $\sparn$
to $\arn$. For this purpose for every action $x\in\act{\sepst}$
we define 
$\Splift_x : \strseta(\sparn) \to \strseta(\arn)$ as
\[
\Splift_x(\sigma)=\sigma\circ \splift_x\enspace.
\]
$\Splift_x(\sigma)$ is a strategy in the arena $\arn$ because
for every state $s_x\in\spstates$, 
the same actions are available in state $s_x$ of the arena $\sparn$
and  in the state $s$ of the arena $\arn$.
\section{The general one-to-two theorem}
\label{sec:principaltheo}

The general formulation of our theorem concerns classes of games
that have three properties: every subgame or split of a game of the class should also be in the class. Moreover in every game in the class the set of strategies available for the players should contain all the deterministic stationary strategies.

\begin{definition}\label{defi:closedgames}
Let $\mathcal{G}=(\arn_i,\strseta_i,\strsetb_i,\pref_i)_{i\in I}$ be a collection of games.
We say that $\mathcal{G}$:
\begin{itemize}
\item
is \emph{subgame-closed} if every subgame of a game in $\mathcal{G}$ is also in $\mathcal{G}$.
\item
is \emph{split-closed} if for every game $G=(\arn,\strseta,\strsetb,\pref)$ in $\mathcal{G}$, every state $\omega$
of $\arn$ there is a game $\cv{G}=(\cv{\arn},\cv{\strseta},\cv{\strsetb},\pref)$ in $\mathcal{G}$,
such that $\sparn$ is the split of $\arn$ on $\sepst$ and for every action $x$ available in $\omega$,
\begin{align}\label{eq:liftsplit}
&\Spproj(\strseta ) \subseteq \cv{\strseta}
\text { and } 
\Spproj(\strsetb) \subseteq \cv{\strsetb} \enspace,\\
&
\notag\Splift_x(\cv{\strseta} ) \subseteq \strseta
\text { and } 
\Splift_x( \cv{\strsetb}) \subseteq \strsetb \enspace.
\end{align}
\item
\emph{contains all deterministic stationary strategies} if in every game $G=(\arn,\strseta,\strsetb,\pref)$ in $\mathcal{G}$,
every deterministic stationary strategy in the arena $\arn$ is contained in $\strseta \cup \strsetb$. 
\end{itemize}
\end{definition}

The following theorem reduces the problem of the existence of \osd\
strategies from two-player games to
one-player games.

\begin{theorem}[The one-to-two theorem]
\label{theo:main}
Let $\mathcal{G}$ be a subgame-closed and split-closed collection of games,
which contains all deterministic stationary strategies.
If \osd\ strategies exist in every one-player game
of $\mathcal{G}$ then \osd\ strategies exist in every two-player game of $\mathcal{G}$.
\end{theorem}

In Section~\ref{sec:weak} are presented two specializations of Theorem~\ref{theo:main} to deterministic games (Theorem~\ref{theo:deterministic}) and games with perfect-information (Theorem~\ref{theo:weak}). In both cases a transitive relation $\pref$ is fixed, on either $\irewards$ for deterministic games or $\mesure(\ihist,\field(\ihist))$ for games with perfect-information. Both results follow from an application of Theorem~\ref{theo:main}
to the class of all games compatible with $\pref$ (all deterministic games for Theorem~\ref{theo:deterministic} and all games with perfect-information for Theorem~\ref{theo:weak}).

\newcommand{\spG}{\cv{G}}

\section{From one-player games to two-player games -- the proof}
\label{sec:mainsec}

This section is devoted to the proof of Theorem~\ref{theo:main}.
We start with some trivial though crucial properties of projections and liftings.

\begin{proposition}[Properties of liftings and projections]
\label{prop:liftproj}
Let $G=(\arn,\strseta,\strsetb,\pref)$ 
and
$\cv{G}=(\cv{\arn},\cv{\strseta},\cv{\strsetb},\pref)$
be two games on arenas $\arn$ and $\sparn$
such that $\sparn$ is the split of $\arn$ on some separating state $\sepst$.
Let $x\in\act{\sepst}$ an action available in  $\sepst$,
let $\sigma\in \strseta$ and $\tau\in \strsetb$ some strategies in $G$
and $\cv{\sigma}\in \cv{\strseta}$ and $\cv{\tau}\in \cv{\strsetb}$ some strategies in $\cv{G}$.

Lifting and projection preserve outcomes:
for every state $s\in\states$,
\begin{align*}
&\Outcome(\arn,s,\sigma,\tau)
=
\Outcome(\sparn,s_x,\Spproj(\sigma),\Spproj(\tau))\\
&\Outcome(\sparn,s_x,\cv{\sigma},\cv{\tau})=
\Outcome(\arn,s,\Splift_x(\cv{\sigma}),\Splift_x(\cv{\tau}))
\enspace.
\end{align*}

Assume moreover that 
condition~\eqref{eq:liftsplit} of Definition~\ref{defi:closedgames} holds, i.e.
\begin{align*}
&\Spproj(\strseta ) \subseteq \cv{\strseta}
\text { and } 
\Spproj(\strsetb) \subseteq \cv{\strsetb} \enspace,\\
&
\notag\Splift_x(\cv{\strseta} ) \subseteq \strseta
\text { and } 
\Splift_x( \cv{\strsetb}) \subseteq \strsetb \enspace.
\end{align*}
If $\cv{\sigma}$ and $\cv{\tau}$ are optimal strategies in $\spG$ from state $s_x$
then 
$\Splift_x(\cv{\sigma})$ and $\Splift_x(\cv{\tau})$
are optimal in $G$ from state $s$.
If $\cv{\sigma}$ is deterministic and stationary and $x=\cv{\sigma}(\sepst)$ then also $\Splift_x(\cv{\sigma})$ 
is deterministic and stationary.
\end{proposition}
\begin{proof}
We prove the first part of the proposition.
Define
\[
\spproj_\infty:\ihist(\sparn,s_x,\cv{\sigma},\cv{\tau})
\to
\ihist(\arn,s,\Splift_x(\cv{\sigma}),\Splift_x(\cv{\tau}))
\]
as
$\spproj^\infty(s_1a_1s_2\cdots)
= 
\spproj(s_1)a_1\spproj(s_2)\cdots $.
Equivalently,
$\spproj^\infty(h)=\lim_n \spproj(h_n)$ where $h_n$ is the prefix of $h$ of length $2n+1$. 
First we show
\begin{equation}
 \label{eq:toto30}
 \prb{\arn,s}{\Splift(\cv{\sigma})}{\Splift(\cv{\tau})}
=
\spproj_\infty \prb{\sparn,s_x}{\cv{\sigma}}{\cv{\tau}}
\end{equation}
Let $h\in\hist(\arn,s,\Splift_x(\cv{\sigma}),\Splift_x(\cv{\tau}))$,
remember that $h^+\subseteq \hist^\infty(\arn,s)$
denotes the cylinder of infinite histories starting with $h$.
According to
Proposition~\ref{lem:imageco},
\[
\spproj_\infty^{-1}(h^+)
= (\splift_x(h))^+ \enspace.
\]
By definition of 
$\prb{\arn,s}{\Splift(\cv{\sigma})}{\Splift(\cv{\tau})}$
and since $\spproj\circ\splift_x$ is the identity on $\hist(\arn)$,
an easy induction shows that
$\prb{\arn,s}{\Splift(\cv{\sigma})}{\Splift(\cv{\tau})}(h^+) = \prb{\sparn,s_x}{\cv{\sigma}}{\cv{\tau}}((\splift_x(h))^+) $.
Finally
$\prb{\sparn,s_x}{\cv{\sigma}}{\cv{\tau}}(\spproj_\infty^{-1}(h^+))=\prb{\arn,s}{\Splift(\cv{\sigma})}{\Splift(\cv{\tau})}(h^+)$.
Since this holds for every cylinder $h^+$ with 
$h\in\hist(\arn,s,\Splift_x(\cv{\sigma}),\Splift_x(\cv{\tau}))$
 we get~\eqref{eq:toto30}. 
 
The projection $\spproj_\infty$ does not change the sequence of rewards thus $\rew\circ \spproj_\infty = \rew$
and
  \[
\Outcome(\arn,s,\Splift_x(\cv{\sigma}),\Splift_x(\cv{\tau}))
  =
\rew \prb{\arn,s}{\Splift_x(\cv{\sigma})}{\Splift_x(\cv{\tau})}
=
\rew \spproj_\infty
 \prb{\sparn,s_x}{\cv{\sigma}}{\cv{\tau}}
=
\rew \prb{\sparn,s_x}{\cv{\sigma}}{\cv{\tau}}
=
\Outcome(\sparn,s_x,\cv{\sigma},\cv{\tau})
\]
which proves 
$\Outcome(\sparn,s_x,\cv{\sigma},\cv{\tau})=
\Outcome(\arn,s,\Splift_x(\cv{\sigma}),\Splift_x(\cv{\tau}))$.
We can apply this equality to $ \cv{\sigma}=\Spproj(\sigma)$ and $\cv{\tau}=\Spproj(\tau)$ and since $\Splift_x\circ \Spproj$ is the identity 
(Proposition~\ref{lem:imageco})
this implies
 $
\Outcome(\sparn,s_x,\Spproj(\sigma),\Spproj(\tau))=
\Outcome(\arn,s,\sigma,\tau)
$ which terminates the proof of the first part of the proposition.

We prove the second part of the proposition.
Assume $\cv{\sigma}$ and $\cv{\tau}$ are optimal in $\sparn$ from state $s_x$.
Let $\sigma\in\strseta$ and $\tau\in\strsetb$ be any strategies in $G$. 
We prove that
\begin{equation}\label{eq:toto}
\Outcome(\arn,s,\sigma,\Splift_x(\cv{\tau}))
\pref
\Outcome(\arn,s,\Splift_x(\cv{\sigma}),\Splift_x(\cv{\tau}))
\pref
\Outcome(\arn,s,\Splift_x(\cv{\sigma}),\tau)\enspace.
\end{equation}
By symmetry, it is enough to prove the left inequality:
\begin{align*}
\Outcome(\arn,s,\sigma,\Splift_x(\cv{\tau}))
&=
\Outcome(\arn,s,\Splift_x(\Spproj(\sigma)),\Splift_x(\cv{\tau}))\\
&=
\Outcome(\sparn,s_x,\Spproj(\sigma),\cv{\tau}))\\
&
\pref
\Outcome(\sparn,s_x,\cv{\sigma},\cv{\tau}))\\
&=
\Outcome(\arn,s,\Splift_x(\cv{\sigma}),\Splift_x(\cv{\tau}))\enspace.
\end{align*}
The first equality holds because $\spproj\circ\splift_x$ is the identity on $\hist(\arn)$
thus $\Splift_x\circ\Spproj$ is the identity on $\Sigma$.
The second equality was established in the first part of the proposition.
The inequality holds because $(\cv{\sigma},\cv{\tau})$ are optimal in $\sparn$
and $\Spproj(\sigma)$ is a strategy in $\spG$ since $\Spproj(\strseta ) \subseteq \cv{\strseta}$ by hypothesis.
The last equality was established in the first part.

Since by hypothesis $\Splift_x(\cv{\strseta} ) \subseteq \strseta$ and 
$\Splift_x( \cv{\strsetb}) \subseteq \strsetb$ then $\Splift_x(\cv{\sigma})$ and $\Splift_x(\cv{\tau}))$ are strategies in $\spG$
and according to~\eqref{eq:toto} they are optimal in $\spG$.

Assume now that $\cv{\sigma}$ is deterministic and stationary and $x=\cv{\sigma}(\sepst)$.
On $\hist(\arn_x)$ the lifting $\splift_x$ is very simple,
it sends every finite history $s_1a_1s_2\cdots s_n\in \hist(\arn_x)$
to
$\splift_x(s_1)a_1\splift_x(s_2)\cdots \splift_x(s_n)$.
As a consequence, a simple induction shows that
every history consistent with $\Splift(\cv{\sigma})$ is a history of $\arn_x$
and
$
\Splift(\cv{\sigma})(s_1a_1s_2\cdots s_n) = \cv{\sigma}(\splift(s_n))$.
Thus $\Splift(\cv{\sigma})$ is deterministic stationary.
\end{proof}

Remark that the last property stated in Proposition~\ref{prop:liftproj} is specific to the case where $x=\cv{\sigma}(\sepst)$,
because in general projection turns deterministic strategies into deterministic strategies
but does not preserve stationarity. For example in the game of Fig.~\ref{fig:splitexample},
the stationary strategy $\sigma$ in $\sparn$ which plays $b$ in $s_a$ and $a$ in $s_b$
is not stationary anymore once projected to $\arn$: when playing with $\Splift_\sepst(\sigma)$, it is necessary to remember which action
was taken on the last visit to $\sepst$.

The proof of Theorem~\ref{theo:main} is performed by induction on the size of the games in $\mathcal{G}$.
The size
of an arena $\arn=\arena$
is defined as 
\[
\size(\arn)=\sum_{s\in\states} (\abs{\act{s}} - 1)
.
\] 
Note that $\size(\arn)\geq 0$ since each state has at least one
available 
action.
By extension, the size of a game is the size of its arena.

The core of the inductive step in the proof of Theorem~\ref{theo:main} is:
\newcommand{\stratsM}{\strseta}
\newcommand{\stratsm}{\strsetb}
\newcommand{\spstratsM}{\cv{\stratsM}}
\newcommand{\spstratsm}{\cv{\stratsm}}

\begin{lemma}\label{lem:inductivestep}
Let $G=(\arn,\strseta,\strsetb,\pref)$ and $\spG=(\sparn,\spstratsM,\spstratsm,\pref)$ be two games
such that $\sparn$ is the split of $\arn$ on some state $\sepst$ of $\arn$ controlled by $\M$.
Assume that:
\begin{enumerate}
\item[i)] at least two actions are available in state $\sepst$ in arena $\arn$,
\item[ii)]
there exists \osd\ strategies in every subgame $G'$ of $\spG$
such that such that $\size(G') < \size(G)$,
\item[iii)]
there exists \osd\ strategies in every subgame $G'$ of $\spG$
which is a one-player game,
\item[iv)]
every deterministic stationary strategy in the arena $\sparn$ is a strategy in the game $\spG$,
i.e. $\spstratsM \cup \spstratsm$ contains all these strategies.
\end{enumerate}
Then there exists a pair $(\covsopt,\covtopt)$ of \osd\ strategies in $\spG$.
\end{lemma}
\begin{proof}
We set the usual notations:
\[
\arn=\arena \text{ and }
\sparn=\sparena\enspace.
\]

For each action $x\in\act{\sepst}$ let
$\cv{\arn}_x$ be the subarena of $\covarena$ induced by the set 
$\cv{\trans}_x$ of transitions:
\begin{equation}
\cv{\trans}_x=\{(\sepst,x,s_x)\mid (\sepst,x,s)\in \trans\}
\cup
\{ (s_x,a,s'_x)\mid \text{$(s,a,s')\in\trans$ and $s\neq\sepst$} \}.
\label{eq:tx}
\end{equation}
For every $x\in\act{\sepst}$ let $\cv{G}_x= (\cv{\arn}_x,\strseta_x,\strsetb_x,\pref_x)$
be the subgame of $\cv{G}$ induced by the subarena $\cv{\arn}_x$ of $\cv{\arn}$.

We will make use several times of the following lemma,
which relates the outcomes in the games $\spG$ and $\spG_x$
when only the action $x$ is played in $\sepst$.
\begin{lemma}\label{lem:spGx}
Let $\sigma$ and $\tau$ be two strategies in $\spG$ and $x\in\act{\sepst}$.
Assume that $\sigma$ is deterministic stationary and
$
\sigma(\sepst)=x
$.
Let $\sigma_x$ and $\tau_x$ be the restrictions
of $\sigma$ and $\tau$ to $\hist(\sparn_x)$.
Then $\sigma_x$ and $\tau_x$ are strategies in $\spG_x$.
Moreover, for every state $s\in\states$,
\[
\Outcome(\sparn,s_x,\sigma,\tau)
=
\Outcome(\sparn_x,s_x,\sigma_x,\tau_x)\enspace.
\]  
\end{lemma}
\begin{proof}
First, remark that $\sigma_x$ and $\tau_x$ are strategies in $\sparn_x$.
For $\sigma_x$ it is obvious: $\sigma_x$ is the deterministic stationary strategy
in $\spG_x$ which associates $\sigma(s_x)$ to $s_x\in S_x$
and in particular $\sigma_x(\sepst) = \sigma(\sepst) = x$.
For $\tau_x$ it is quite obvious as well, because every state $s_x\in S_x$ controlled by 
player $\m$ is different from $\sepst$ thus the same actions are available in $s_x$
in both arenas $\sparn$ and $\sparn_x$.

Therefore, $\sigma_x$ and $\tau_x$ are strategies in the subgame $\spG_x$
by definition of a subgame (Definition~\ref{defi:subgame}).

According to the separation property (Proposition~\ref{prop:sep}),
since $\sigma(\sepst) =x$
every history consistent with $\sigma$ and starting from $S_x$ stays in $\sparn_x$,
thus $\hist(\sparn,s_x,\sigma,\tau)=\hist(\sparn_x,s_x,\sigma_x,\tau_x)$.
Since $\sigma$ and $\sigma_x$ on one hand and $\tau$ and $\tau_x$ on the other hand coincide on
$\hist(\sparn_x,s_x,\sigma_x,\tau_x)$ then
$\prb{\sparn,s_x}{\sigma}{\tau}
=
\prb{\sparn_x,s_x}{\sigma_x}{\tau_x}$.
By definition of outcomes, this terminates the proof of Lemma~\ref{lem:spGx}.
\end{proof}

Note that, modulo the renaming of states,
$\cv{\arn}_x$ is isomorphic to the subarena $\arn_x$ of $\arn$ 
obtained by removing from $\act{\sepst}$ all actions except $x$.
According to hypothesis i) $|\act{\sepst}|\geq 2$
thus
$\size(\cv{\arn}_x) < \size(\arn)$.
According to hypothesis ii), in the game $\cv{G}_x$ there are \osd\ strategies
\[
(\covsopt_x,\covtopt_x)\in\strseta_x\times\strsetb_x\enspace.
\] 
Given these strategies,
our aim is to construct a pair $(\covsopt,\covtopt)$
of \osd\ strategies in the game $\spG$.

\paragraph{Construction of $\covtopt$}
Let 
$\covtopt\in\strsetb(\covarena)$ 
be the deterministic stationary strategy of player \m\ in $\sparn$
 defined  in the following way:
for each $s\in\states$ such that $\player(s)=\m$ and all
  $x\in\act{\sepst}$,
\begin{equation}
\covtopt(s_x)=\covtopt_x(s_x)\enspace.
\label{eq:defcovtopt}
\end{equation}

\paragraph{Construction of $\covsopt$}
The construction of $\covsopt$ simply consists in selecting one of the \osd\
strategies $(\sigma_x)_{x\in\act{\sepst}}$
in $\spG_x$ and extend it in a natural way to $\spG$.

The point is to choose the right action $x\in\act{\sepst}$.
For that we consider the subarena $\covarena[\covtopt]$ of $\covarena$ 
obtained from \covarena\ by removing all transitions $(s_x,a,s'_x)$
controlled by \m\ and such that $a\neq \covtopt(s_x)$.
In other words, we restrain the moves of player \m\ 
in $\covarena[\covtopt]$ by forcing him
to play actions according to  strategy \covtopt.
Therefore $\covarena[\covtopt]$ is a subarena of $\covarena$.
We denote $\cv{G}[\covtopt]$ the subgame of $\cv{G}$
induced by $\covarena[\covtopt]$.

Note that  only player \M\ has  freedom to choose
actions in $\covarena[\covtopt]$, 
thus $\cv{G}[\covtopt]$ is a one-player game
controlled by player \M.
Therefore, according to hypothesis iii)
there are \osd\ strategies
$
(\zeta^\opt,\tau_0)
$
in the game $\cv{G}[\covtopt]$.
Actually $\tau_0$ is the unique trivial strategy of player $\m$
in $\covarena[\covtopt]$.

Since there is no restriction on actions of player $\M$ in
$\covarena[\covtopt]$ then $\zeta^\opt$ is also a strategy
in arena $\covarena$ and by definition of a subgame,
$\zeta^\opt$ is a strategy in the game $\spG$:
\begin{equation}\label{eq:zetaok}
\zeta^\opt \in \cv{\strseta}\enspace.
\end{equation}


Let $e\in\act{\sepst}=\actf{\cv{\trans}}{\sepst}$ be
the action chosen by $\zeta^\opt$ in the state $\sepst$,
\begin{equation}
\zeta^\opt(\sepst)=e\enspace.
\label{eq:actchoice}
\end{equation}

The stationary and deterministic strategy $\covsopt:\spstates\to \actions$ is the extension
of the \osd\ strategy $\covsopt_e:\spstates_e\to A$ from $\spstates_e$ to $\spstates$ defined by:
\[
\forall s_x\in\spstates, \covsopt(s_x) = \covsopt_e(s_e)\enspace.
\]
In particular, since $e$ is the only action available in $\sepst$
in the arena $\cv{\arn}_e$,
\begin{equation}\label{eq:covsopte}
\covsopt(\sepst) = e\enspace.
\end{equation}

\paragraph{Strategies $(\covsopt,\covtopt)$ are optimal in $\spG$.}
According to hypothesis iv), the strategies strategies $\covsopt$ and $\covtopt$ are strategies in $\spG$.
We  show they are optimal in the game $\cv{G}$ if the initial state
belongs to $\states_e$,
i.e. for all strategies 
$\sigma\in\cv{\strseta},\tau\in\cv{\strsetb}$ 
and all $s\in\states$,
\begin{equation}
\Outcome(\covarena,s_e, \sigma, \covtopt)
\pref
\Outcome(\covarena,s_e, \covsopt, \covtopt)
\pref
\Outcome(\covarena,s_e, \covsopt, \tau)\enspace.
\label{eq:opti}
\end{equation}

\medskip 

We start with the proof of the right-handside of~\eqref{eq:opti}.
Let $\tau$ be any strategy for $\m$ in $\spG$.
Then its restriction $\tau_e$ to $\hist(\covarena_e)$
is also strategy in $\sparn_e$ because
for every state $t_e \in S_e$ controlled by player $\m$,
$t_e\neq \sepst$ thus 
exactly the same actions are available in $t_e$ in both arenas $\covarena$ and $\covarena_e$.
The right inequality in~\eqref{eq:opti} decomposes as:
\[
\Outcome(\covarena,s_e, \covsopt, \covtopt)
=
\Outcome(\covarena_e,s_e, \covsopt_e, \covtopt_e)
\pref_e
\Outcome(\covarena_e,s_e, \covsopt_e, \tau_e)
=
\Outcome(\covarena,s_e, \covsopt, \tau)
\enspace.
\]
The left and right equalities holds according to Lemma~\ref{lem:spGx}.
The central inequality holds because $\covsopt_e$ and $\covtopt_e$
are optimal in $\spG_e$ and
$\tau_e$ is a strategy in $\spG_e$, according to Lemma~\ref{lem:spGx} again.
Since $\pref_e$ is a restriction of $\pref$ we get the right-handside of~\eqref{eq:opti}. 

%

\medskip

Now we prove the left-handside
of~\eqref{eq:opti}.
Let $\sigma$ be any strategy for $\M$ in $\spG$.
Then the restriction $\sigma'$ of $\sigma$ to $\hist(\covarena[\covtopt])$
is a strategy in $\sparn[\covtopt]$ because the same actions are available
in every state $s$ controlled 
by $\M$ in both arenas $\covarena$ and $\covarena[\covtopt]$.
Let $\zeta^\opt_e$ be the restriction of
$\zeta^\opt$ to the states of $\states_e$.
%
%
%
The left inequality in~\eqref{eq:opti} decomposes as:
\begin{align}
\label{eq:zeta1}
\Outcome(\covarena,s_e,\sigma,\covtopt)
&=
 \Outcome(\covarena[\covtopt],s_e,\sigma',\tau_0)\\
\label{eq:zeta2}
& \pref 
 \Outcome(\covarena[\covtopt],s_e,\zeta^\opt,\tau_0)\\
\label{eq:zeta3}
& =
 \Outcome(\covarena,s_e,\zeta^\opt,\covtopt)\\
\label{eq:zeta4}
& =
 \Outcome(\covarena_e,s_e,\zeta_e^\opt,\covtopt_e)\\
\label{eq:zeta5}
& \pref_e
 \Outcome(\covarena_e,s_e,\covsopt_e,\covtopt_e)\\
\label{eq:zeta6}
& =
 \Outcome(\covarena_e,s_e,\covsopt,\covtopt)
 \enspace.
\end{align}
The equality~\eqref{eq:zeta1} holds
because $\hist(\covarena,s_e,\sigma,\covtopt)=\hist(\covarena[\covtopt],s_e,\sigma',\tau_0)$
and for every finite history $h\in\hist(\covarena,s_e,\sigma,\covtopt)$, 
$\prb{\covarena,s_e}{\sigma}{\covtopt}(h^+)=\prb{\covarena[\covtopt],s_e}{\sigma'}{\tau_0}(h^+)$.
The inequality~\eqref{eq:zeta2} holds because $\zeta^\opt$ and $\tau_0$ are optimal in $\cv{G}[\covtopt]$
and by definition of subgames $\sigma'$ is a strategy in the subgame $\cv{G}[\covtopt]$ and the preference order in $\spG[\covtopt]$ is a restriction of $\pref$.
The equality~\eqref{eq:zeta3} holds for similar reasons than~\eqref{eq:zeta1} does.
The equalities~\eqref{eq:zeta4} and~\eqref{eq:zeta6}
are consequences of Lemma~\ref{lem:spGx}.
The inequality~\eqref{eq:zeta5} holds because $(\covsopt_e,\covtopt_e)$ is a pair of optimal strategies in $\spG_e$
and $\zeta_e^\opt$ is a strategy in $\spG_e$ according to Lemma~\ref{lem:spGx} again.
Since $\pref_e$ is a restriction of $\pref$ we get the left inequality in~\eqref{eq:opti}.

Finally the two inequalities in~\eqref{eq:opti} do hold and since $\covsopt$ is deterministic stationary the proof of Lemma~\ref{lem:inductivestep} is over.
\end{proof}

\begin{proof}[Proof of Theorem~\ref{theo:main}]
We fix a subgame-closed and split-closed collection $\mathcal{G}$ of games which contains all deterministic stationary strategies.
We denote $\mathcal{G}=(G_i)_{i\in I}$
and for every $i\in I$ we denote
\[
G_i = (\arn_i,\strseta_i,\strsetb_i,\pref_i)_{i\in I}\enspace.
\]



The proof of Theorem~\ref{theo:main}
is carried out by induction on the size of $\arn_i$.

The case where $\size(\arn_i)=0$ for some $i\in I$ is trivial.
In this case there is a unique action available in each state $s$ of $\arn_i$,
and each player has  a unique strategy which is \osd.

Let $G=(\arn,\strseta,\strsetb,\pref)$ a game in $\mathcal{G}$
such that $\size(G)>0$
and suppose  
that the theorem holds for all games $G_i$ in $\mathcal{G}$
whose arenas $\arn_i$ satisfy $\size(\arn_i) <  \size(\arn)$.

In case $G$ is a one-player game, then by hypothesis \osd\ strategies exist
for both players. In the sequel we assume $G$ is \emph{not}
a one-player game.

Therefore there exists 
a state $\sepst$ in $\arn$ such that $\act{\sepst}\geq 2$.
Assume first that 
\begin{equation}\label{eq:control}
\player(\sepst)=\M\enspace.
\end{equation}

Let $\sparn$ be the split of $\arn$ on $\sepst$.
Then, since $\mathcal{G}$ is split-closed there exists a game
$\cv{G}=(\sparn,\cv{\strseta},\cv{\strsetb},\pref)$ in $\mathcal{G}$
such that property~\eqref{eq:liftsplit} of Definition~\ref{defi:closedgames} holds for every action $x$ available in $\sepst$.

All conditions of Lemma~\ref{lem:inductivestep} are satisfied for the games
$G$ and $\spG$: condition i) is by choice of $\sepst$,
condition ii) is the inductive hypothesis,
and by hypothesis conditions iii) and iv) hold for every game in $\mathcal{G}$, in particular for $\spG$.
As a consequence there exists a pair
$(\covsopt,\covtopt)$ of \osd\ strategies in $\spG$.

Let $e=\covsopt(\sepst)$ be the action played by $\covsopt$ in the separation state
 and $\Splift_e(\covsopt)$ the lifting of $\covsopt$ from $\sparn$ to $\arn$ with respect to action $e$.
By choice of $\spG$, condition~\eqref{eq:liftsplit} holds thus
$\Splift_x(\cv{\strseta} ) \subseteq \strseta$ 
and $\Splift_e(\covsopt)$ is a strategy in $G$.
According to the second part of Proposition~\ref{prop:liftproj},
the strategy $\Splift_e(\covsopt)$ is both
optimal in $G$
and deterministic stationary.

Therefore $\M$ has an \osd\ strategy in $G$.

To find an \osd\ for player \m\ in $G$
it suffices to choose as a separation state a state controlled
by player \m\ with at least two actions available.
Such a state exists because $G$ is not a one-player game.
By a reasoning symmetric to the one developped previously we
can construct another pair of optimal strategies 
$(\sigma^\star,\tau^\star)$ in $G$,
however now the strategy $\tau^\star$ of player \m\ will be
deterministic stationary.

But in zero-sum games if we have two pairs of optimal strategies
$(\sigma^\opt,\tau^\opt)$ and $(\sigma^\star,\tau^\star)$ then 
$(\sigma^\opt, \tau^\star)$ is also a pair of optimal strategies.

This ends the proof of Theorem~\ref{theo:main}.

\medskip

Even if it has no bearing on the proof,
we provide an explicit description
of the strategy $\Splift_e(\covtopt)$.
Let $h\in\hist(\arn)$ 
be any finite history consistent with \arn\ starting in a state $s$ and
ending in a state $s$ controlled by \m\
and $\cv{h} =\phi_e(h)$ the lifting of $h$ in 
$\hist(\covarena)$ with respect to $e$.
Then $\cv{h}$ ends in a state $s_x\in\iembed(s)$ for some
$x\in\act{\sepst}$. By definition, 
$\plcd{\ipembed(\covtopt)}{h}= \plcd{\covtopt}{\cv{h}}=\plcd{\covtopt}{s_x}=
\plcd{\covtopt_x}{s_x}$. Informally, for a history ending in a state $s$
player \m\ chooses one of the actions $\plcd{\covtopt_x}{s_x}$,
$x\in\act{\sepst}$.

The problem is to see which of
these actions should be chosen and this depends on the last
state of $\cv{h}$. Thus
the question is if we can obtain the last state of $\cv{h}$
without explicitly calculating the whole lifted history $\cv{h}$?
If $h$ never visits the state $\sepst$
then, since $\cv{h}$ begins in a state of $\covstates_e$ this history can only end
in the state $s_e\in\covstates_e$, i.e 
$\plcd{\ipembed(\covtopt)}{h}=\topt_e(s_e)$.
If $h$ visits $\sepst$ then all depends on the action chosen by \M\
during the last visit to \sepst, if this action is $x\in\act{\sepst}$
then the last state of $\cv{h}$ is $s_x$ and therefore
$\plcd{\ipembed(\covtopt)}{h}=\covtopt_x(s_x)$. Notice that the strategy
$\ipembed(\covtopt)$ is not stationary, however we only need a finite
amount of memory to implement it, 
 it is sufficient to 
know all strategies $\covtopt_x$, the last state $s$ of $h$,
whether since the begining of the game  the state \sepst\ was visited
or not
and if it was visited then 
what action $x\in\act{\sepst}$ was taken during the last visit to \sepst.
\end{proof}

\section{Final remarks}

Theorem~\ref{theo:main} gives a sufficient condition for the existence
of \osd\  strategies for a given two-player
game.
%

The examples  in Section~\ref{sec:payoff} were given only to
illustrate the method, we do not think that reestablishing
known results is of
particular interest.

Let us note that in recent years, 
attempting to capture   subtle 
aspects of computer systems behaviour, 
several new games were proposed,
as typical example 
we can cite~\cite{CHJ05}
combining parity and mean-payoff games.
We hope that Theorem~\ref{theo:main} and its deterministic counterpart
Theorem~\ref{theo:deterministic} 
will prove useful in the study of such new games.

Finally, note that an improved version of the results is in preparation, for a class of games both more general
and simpler, where probabilities are abstracted as non-determinism.

\bibliographystyle{alpha}
\bibliography{2j1j}

\end{document}